\documentclass[twocolumn,english,amsthm]{autart}   

\usepackage{graphicx}          
\usepackage{epstopdf}          

\usepackage{babel}
\usepackage{natbib}
\usepackage{amsmath,amssymb}
\usepackage{booktabs}
\usepackage[shortcuts]{extdash}

\allowdisplaybreaks

\DeclareMathOperator*{\argmin}{arg\;min}

\newtheorem{mydef}{Definition}

\newtheorem{assumption}{Assumption}

\begin{document}

\begin{frontmatter}

\title{Extending the Best Linear Approximation Framework to the Process Noise Case} 

\thanks[footnoteinfo]{The corresponding author is M.~Schoukens (m.schoukens@tue.nl).}

\author{Maarten Schoukens$^1$, Rik Pintelon$^2$, Tadeusz P. Dobrowiecki$^3$, Johan Schoukens$^2$} 
\address{$^1$ Eindhoven University of Technology, Control Systems Group,  Eindhoven, The Netherlands}  
\address{$^2$ Vrije Universiteit Brussel, Dept. ELEC, Pleinlaan 2, B-1050 Brussels, Belgium}  
\address{$^3$ Budapest University of Technology and Economics, Department of Measurement and Information Systems, H-1117 Budapest, Hungary}  
          

\begin{abstract}                          
	The Best Linear Approximation (BLA) framework has already proven to be a valuable tool to analyze nonlinear systems and to start the nonlinear modeling process. The existing BLA framework is limited to systems with additive (colored) noise at the output. Such a noise framework is a simplified representation of reality. Process noise can play an important role in many real-life applications.
	
	This paper generalizes the Best Linear Approximation framework to account also for the process noise, both for the open-loop and the closed-loop setting, and shows that the most important properties of the existing BLA framework remain valid. The impact of the process noise contributions on the robust BLA estimation method is also analyzed.
\end{abstract}

\end{frontmatter}

\section{Introduction}
	A linear approximate model of a nonlinear system often offers valuable insight into the linear (but also nonlinear) behavior of that system. The Best Linear Approximation (BLA) framework described in \citep{Schoukens1998,Pintelon2002,Enqvist2005,Enqvist2005a,Schoukens2009,Pintelon2012} offers such a well understood and valuable approximation framework for a wide class of practically important signals and systems (see in detail in Section~\ref{sec:BlaOE}). The BLA framework is often used to analyze how nonlinearly a system behaves (see for instance \citep{Vaes2015} for mechanical systems and \citep{Vlaar2017} for biomechanical systems), to guide the user to select a good nonlinear model structure \citep{Schoukens2015}, to obtain linear models in the presence of nonlinearities \citep{Pintelon2012}, but also to start the estimation process of a nonlinear model \citep{Paduart2010,SchoukensM2017c,SchoukensM2017b}. A recent overview of the BLA framework, and its use in practical applications (e.g. ground vibration testing, combustion engine, robotics and electronics), is provided in \citep{Schoukens2016}.
	
	The BLA framework was initially defined for systems operating in open loop and with additive (colored) noise present at the output only. The extension towards the closed-loop setting has been made in \citep{Pintelon2012a,Pintelon2013}. The extension of the BLA framework to include process noise is the subject of this paper.
	
	Considering additive (colored) noise at the output only is a simplified representation of reality. This simplification can lead to biased nonlinear model estimates when other noise sources are present, located at other positions inside the system, e.g. process noise passing through a nonlinear subsystem \citep{Hagenblad2008}. A more realistic noise framework can be obtained by introducing multiple noise sources, or by placing the noise source at another location in the model structure. It is important to offer the user a theoretical framework with which (s)he can analyze the influence of process noise on the system, and whether or not it is important to include process noise in the nonlinear modeling step. An extended BLA framework could offer this theoretical insight. An initial step towards a generalized BLA analysis is made in \citep{Giordano2016} where the BLA of a Wiener-Hammerstein system is analyzed in the presence of process noise.
	
	This paper introduces first the classical open-loop BLA framework with additive noise at the system output only (Section~\ref{sec:BlaOE}). The extension towards systems with process noise is presented in Section~\ref{sec:BlaPN}. It is shown in Section~\ref{sec:FeedbackPN} that this extension is also valid for systems operating in closed loop. The impact of the process noise contributions on the robust BLA estimation method is analyzed in Section~\ref{sec:Estimation}. Finally, the proposed process noise BLA framework is illustrated in Section~\ref{sec:Simulation}. A discrete-time setting is used throughout the paper. However, all derivations and proofs can easily be generalized to the continuous-time setting.
	
\section{Best Linear Approximation: Additive Noise at the System Output} \label{sec:BlaOE}	
		
		\begin{figure}
			\centering
				\includegraphics[width=0.55\columnwidth]{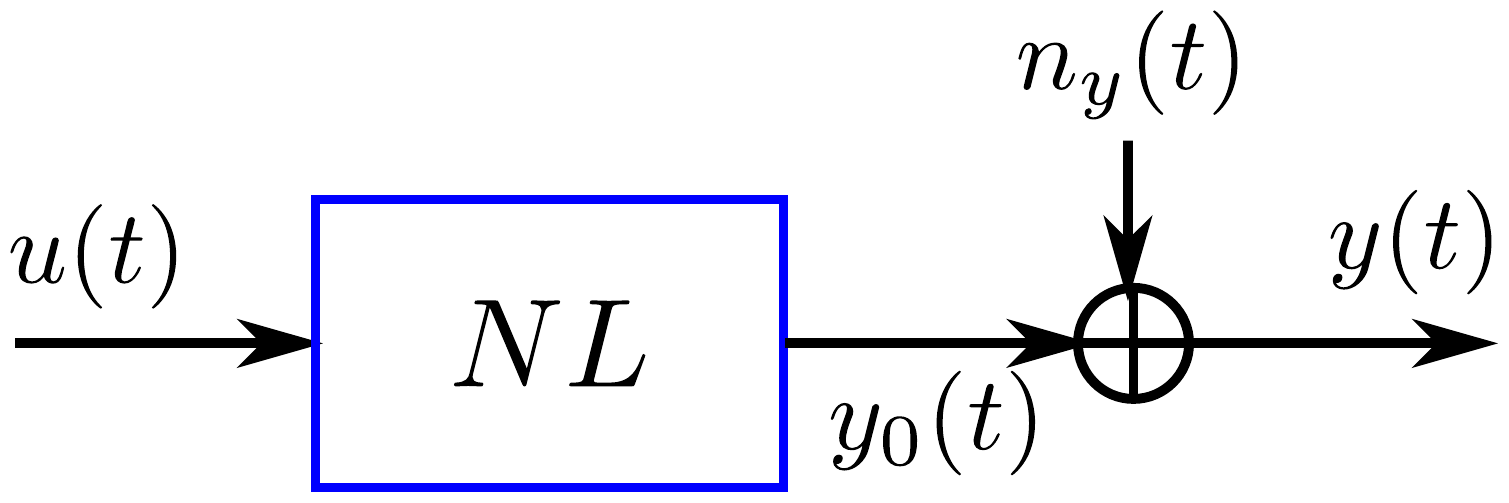}
			\caption{A dynamic nonlinear system with zero-mean (colored) additive noise $n_y(t)$ at the output only. The input excitation $u(t)$ belongs to the signal class $\mathbb{U}$.}
			\label{fig:NonlinearOutputNoise}
		\end{figure}
		
		\subsection{Signal Class}
			This paper assumes the input signal $u(t)$ (see Figure~\ref{fig:NonlinearOutputNoise}) to belong to a generalization of the Gaussian signal class: the Riemann equivalence class of asymptotically normally distributed excitation signals \citep{Schoukens2009,Pintelon2012}. Note that the BLA can also be defined for other signal classes. This choice has been made here to follow the framework defined in \citep{Pintelon2012}, and to use the same class of systems. However, many of the properties that are derived in this paper are not limited to the chosen input signal class. The impact of the input signal class is studied in \citep{Wong2012}.
			
			\begin{mydef} \label{def:Gaussian}
				Riemann equivalence class of asymptotically normally distributed excitation signals $\mathbb{U}$.
				Consider a signal $u(t)$ with a power spectrum $S_{UU}(j\omega)$. $S_{UU}(j\omega)$ is piecewise continuous, with a finite number of discontinuities. A random signal $u(t)$ belongs to the Riemann equivalence class if it obeys any of the following statements:
				\begin{enumerate}
					\item $u(t)$ is a Gaussian process with power spectrum $S_{UU}(j\omega)$.
					\item $u(t)$ is a random multisine or random phase multisine \citep{Pintelon2012} such that:
				\end{enumerate}					\vspace*{-\baselineskip}
				\begin{align}
					\frac{1}{N} \sum_{k = k_1}^{k_2} E\left\{ \left| U\left( j\omega_k \right) \right|^2 \right\} &= \frac{1}{2 \pi} \int_{\omega_{k_1}}^{\omega_{k_2}} S_{UU}\left( j \nu \right) d\nu + O\left( \frac{1}{N} \right), \nonumber
				\end{align}
				with $\omega_k = k \frac{2 \pi f_s}{N}$, $k \in \mathbb{N}$, $0\leq\omega_{k_1}<\omega_{k_2}<\pi f_s$, $f_s$ is the sampling frequency.

			\end{mydef}
				
			A random phase multisine $u(t)$ is a periodic signal with period length $\frac{N}{f_s}$, where $N$ is the number of samples in one period, defined in \citep{Pintelon2012} as: 
				\begin{align}
					u(t) &= \frac{1}{\sqrt{N}}\sum_{k=0}^{N/2-1}2U_{k}\cos(2\pi k\frac{t}{N}+\varphi_{k}).
				\end{align}
			The phases $\varphi_{k}$ are random variables that are independent over the frequency and are a realization of a random process on $[0, 2\pi[$, such that $E\{e^{j\varphi_{k}}\}=0$.  For instance, the random phases can be uniformly distributed between $[0, 2\pi[$. The (real) amplitude $U_{k}$ is set in a deterministic way by the user. $U_k$ is uniformly bounded by $M_U$ ($0 \leq U_k \leq M_U < \infty$).
			
			Note that the Riemann equivalence class of asymptotically normally distributed excitation signals can, in most cases, easily be tuned to fit the application at hand, without much additional processing or hardware. The random phase multisine signals are periodic excitation signals, offering the opportunity of leakage-free measurements with a full control on the amplitude spectrum.
		
		\subsection{System Class and Noise Framework} \label{sec:SystemClass}
			It is assumed that the nonlinear system output can be represented arbitrarily well in the least-squares sense by a fading memory Volterra kernel representation \citep{Schetzen1980,Boyd1985}. This system class contains, for instance, systems with a hard saturation nonlinearity, but does not contain bifurcating and chaotic systems. The generality of this system class is discussed in detail by \citep{Boyd1985}.
			
			The Volterra model output consists of the sum of the outputs of the kernels of different degree. The output of a Volterra kernel of degree $\alpha$ is given by in the time domain by:
			\begin{align}
				y_{\alpha}(t) &= V_{\alpha}(u(t)), \label{eq:ya} \\
					 &= \sum_{k_1,\ldots,k_\alpha=0}^{N_k} h_{\alpha}(k_1,\ldots,k_\alpha) u(t-k_1)\ldots u(t-k_\alpha), \nonumber
			\end{align}
			which results in the following frequency domain representation:
			\begin{align}
							Y_\alpha(j\omega_k) &= \frac{1}{N^{\frac{\alpha-1}{2}}} \sum_{k_1,k_2,\ldots,k_{\alpha-1}=-N/2+1}^{N/2-1} H^{\alpha}_{L_k,k_1,k_2,\ldots,k_{\alpha-1}} \nonumber \\
													 & \quad \quad \quad U(j\omega_{L_k})U(j\omega_{k_1})\ldots U(j\omega_{k_{\alpha-1}}),
						\end{align}
			where $L_k = k - k_1 - k_2 - \ldots - k_{\alpha-1}$. $H^{\alpha}_{L_k,k_1,k_2,\ldots,k_{\alpha-1}}$ is a symmetrized frequency domain representation of the Volterra kernel $h_{\alpha}$ of degree $\alpha$ \citep{Schetzen1980,Pintelon2012}. $Y_\alpha(j\omega_k)$ is obtained as the Discrete Fourier Transform (DFT) of $y_\alpha(t)$:
			\begin{align}
				Y_\alpha(j\omega_k) = \frac{1}{\sqrt{N}} \sum_{t=0}^{N-1} y_\alpha(t) e^{-j2\pi t k /N}, \\
				y_\alpha(t) = \frac{1}{\sqrt{N}} \sum_{k=0}^{N-1} Y_\alpha(j\omega_k) e^{j2\pi t k/N}.
			\end{align}
			
			\begin{mydef} \label{def:System}
				$\mathbb{S}$ is the class of nonlinear systems such that, when excited by a random phase multisine:
				\begin{align}
					\exists \; C_1 < \infty, \text{s.t.} \; \sum_{\alpha=1}^{\infty} M_{G^\alpha}M^{\alpha}_{U} \leq C_1,
				\end{align}
				with $M_{G^\alpha} = \textit{max}\left|H^{\alpha}_{L_k,k_1,k_2,\ldots,k_{\alpha-1}}\right|$, where the maximum is taken over the indices $L_k,k_1,k_2,\ldots,k_{\alpha-1}$.
			\end{mydef}
			Definition \ref{def:System} postulates the existence of a uniformly bounded fading memory Volterra series whose output converges in mean square sense to the output of the nonlinear system belonging to the system class $\mathbb{S}$ (see \citep{Pintelon2012} for more detail) when the degree $\alpha$ grows to infinity.
		
			\begin{assumption} \label{ass:Noise}
				\textbf{Output noise framework:}
				An additive, colored zero-mean noise source $n_y(t)$ with a finite variance $\sigma_{n_y}^2$ is present at the output of the system (see Figure~\ref{fig:NonlinearOutputNoise}):
				\begin{align}
					y(t) = y_0(t) + n_y(t).
				\end{align}
				This noise $n_y(t)$ is assumed to be independent of the known input $u(t)$. $y(t)$ is the actual measured output signal and a subscript $0$ denotes the exact (noise-free) value.
			\end{assumption}

	\subsection{Definition of the Best Linear Approximation}
			
		The BLA model of a nonlinear system belonging to system class $\mathbb{S}$ with zero-mean (colored) additive noise at the system output only (see Figure~\ref{fig:NonlinearOutputNoise}) is a linear time-invariant (LTI) approximation of the behavior of that system. The BLA is defined in \citep{Schoukens1998,Pintelon2002,Schoukens2009,Pintelon2012} as:
			\begin{align}
		 	  G_{bla}(q) 		&= \underset{G(q)}{\argmin} \: E_{u,n_y}\left\{ \left| \tilde{y}(t) - G(q)\tilde{u}(t) \right|^{2} \right\}, \label{eq:BLA} \\
				\tilde{u}(t)	&= u(t) - E_{u}\left\{ u(t) \right\}, \label{eq:uTilde} \\
		 	  \tilde{y}(t)	&= y(t) - E_{u,n_y}\left\{ y(t) \right\}, \label{eq:yTilde} 
			\end{align}
	 	where $E_{u,n_y}\left\{.\right\}$ denotes the expected value operator taken w.r.t. the random variations due to the input $u(t)$ and the output noise $n_y(t)$ and $G(q)$ belongs to the set of all possible LTI systems.	This definition of the BLA is equivalent to the definition of the linear time-invariant second-order equivalent model defined in \citep{Ljung2001, Enqvist2005, Enqvist2005a} when the stability and causality restrictions imposed there are omitted.
	 	
	 	It is shown that the BLA is given by \citep{Enqvist2005, Enqvist2005a, Pintelon2012}:
	 	\begin{align}
		 	G_{bla}(j\omega) = \frac{S_{YU}(j\omega)}{S_{UU}(j\omega)}, 
	 	\end{align}
	 	where $S_{YU}(j\omega)$ is the cross-power spectrum of $u(t)$ and $y(t)$ and $S_{UU}(j\omega)$ is the autopower spectrum of $u(t)$. Hence, the existence of $G_{bla}(j\omega)$ is guaranteed if $S_{UU}(j\omega)$ and $S_{YU}(j\omega)$ exist, and $S_{UU}(j\omega) \neq 0$. The existence of $S_{YU}(j\omega)$ is guaranteed by the chosen signal (Definition~\ref{def:Gaussian}) and system class (Definition~\ref{def:System}) \citep{Pintelon2012,Schetzen1980}. The BLA is not defined at the frequencies where $S_{UU}(j\omega) = 0$ \citep{Enqvist2005,Pintelon2012}. A (possible infinite order, noncausal) transfer function or impulse response representation $G_{bla}(q)$ can be obtained by fitting $G_{bla}(j\omega)$ at the excited frequencies.
		
		\begin{figure}
			\centering
				\includegraphics[width=0.8\columnwidth]{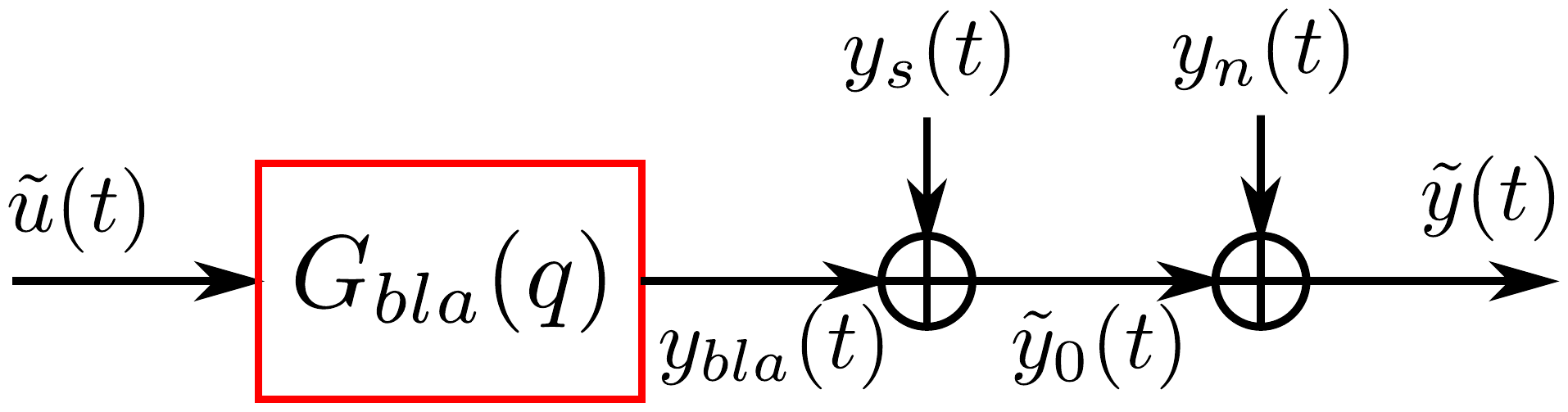}
			\caption{The BLA of a nonlinear system for a given class of excitation signals $\mathbb{U}$ consists of the resulting LTI model $G_{bla}(q)$, the stochastic nonlinear contribution $y_{s}(t)$, and the noise distortion $y_{n}(t)$.}
			\label{fig:BestLinearApprox}
		\end{figure}
		
		Three constituents of the BLA framework can be defined (see Figure~\ref{fig:BestLinearApprox}): the BLA $G_{bla}(q)$ itself, the stochastic nonlinear contribution $y_s(t)$ and the noise contribution $y_n(t)$. The output residuals $y_{t}(t)$ represent the total distortion that is present at the output of the system. The total distortion can be split in two contributions based on their nature as is depicted in Figure~\ref{fig:BestLinearApprox}. The stochastic nonlinear contribution $y_{s}(t)$ represents the unmodeled nonlinear contributions, while the noise contribution ${y}_n(t)$ is the additive noise that is present at the system output:
			\begin{align}
				y_{t}(t)   &= \tilde{y}(t) - G_{bla}(q)\tilde{u}(t) = y_s(t) + y_n(t), \\
		 	  y_{s}(t)   &= \tilde{y}_{0}(t) - G_{bla}(q)\tilde{u}(t), \\
		 	  y_{n}(t)   &= \tilde{y}(t) - \tilde{y}_{0}(t) = n_y(t),
			\end{align}
		where $\tilde{y}_{0}(t) = E_{n_y}\{ \tilde{y}(t)\}$ is the unknown zero-mean noiseless output. The nonlinear distortion $y_{s}(t)$ is linearly uncorrelated with the input $\tilde{u}(t)$ ($E_u\{y_s(t)\tilde{u}(\tau)\} =0 \: \forall \: t,\tau$). The nonlinear distortion $y_{s}(t)$ is not independent of the input $\tilde{u}(t)$ however \citep{Pintelon2012}. The noise distortion $y_{n}(t)$ on the contrary is both uncorrelated with the input and independent of the input $\tilde{u}(t)$. All three signals $y_{bla}(t)$, $y_{s}(t)$, $y_{n}(t)$ are zero-mean.

\section{Best Linear Approximation: Process Noise Extension} \label{sec:BlaPN}

	\subsection{Considered System Class and Noise Framework} \label{sec:VoltPN}
		
		The considered excitation signal class and the output noise assumptions are unchanged with respect to Section~\ref{sec:BlaOE}. The considered system class is extended here to include process noise, and the necessary assumptions on the process noise are formulated.
		
		\begin{figure}
			\centering
				\includegraphics[width=0.55\columnwidth]{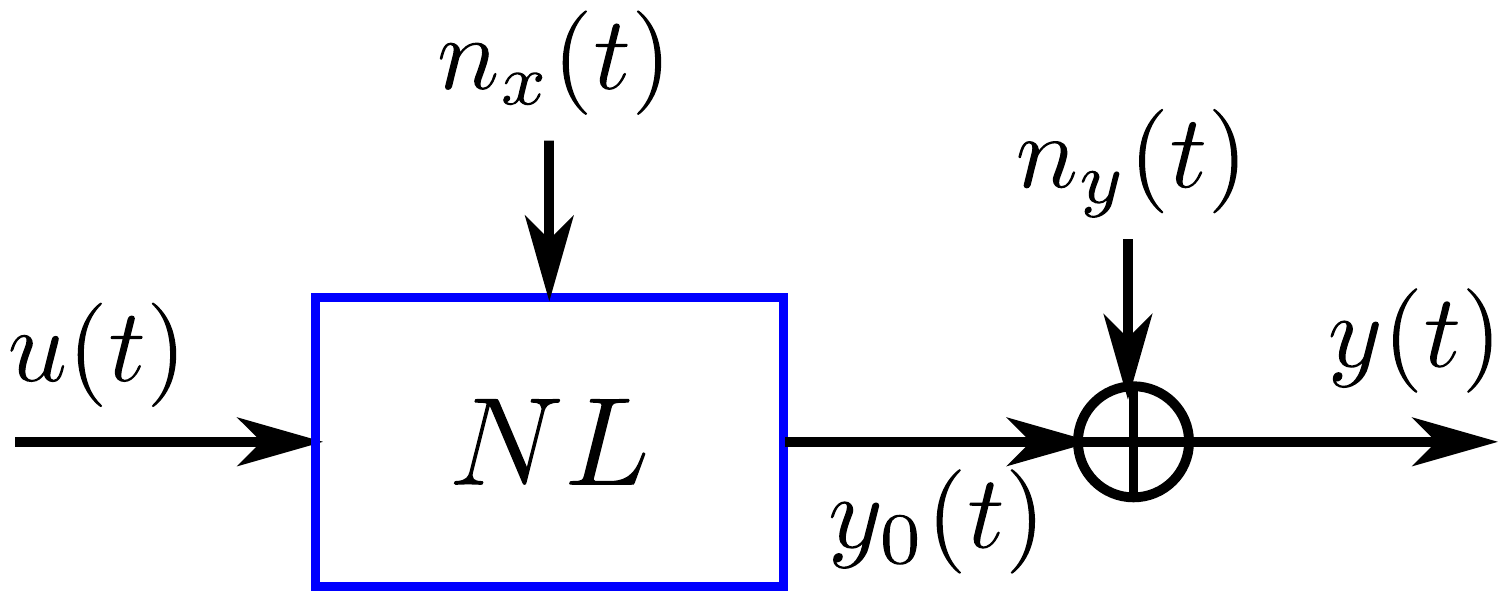}
			\caption{A dynamic nonlinear system with zero-mean (colored) process noise $n_x(t)$ and zero-mean (colored) additive output noise $n_y(t)$. The input excitation $u(t)$ belongs to the signal class $\mathbb{U}$.}
			\label{fig:NonlinearProcessNoise}
		\end{figure}
		
		It is assumed in Section~\ref{sec:SystemClass} that the underlying nonlinear system is a Volterra system. Here we extend the standard single-input single-output Volterra kernel representation to a dual-input single-output representation where one of the inputs is excited by the process noise (see Figure~\ref{fig:NonlinearProcessNoise}). The modeling of the process noise as the second input to the Volterra system describing the measured nonlinear system opens considerable opportunities for a unified treatment of many process noise configurations. The output of a dual-input m,n-th order Volterra kernel with process noise is given by:
		\begin{align}
			y_{m,n}(t) &= V_{m,n}(u(t),n_x(t)), \label{eq:ymn} \\
						 &= \sum_{k_1,\ldots,k_m=0}^{N_k} \sum_{j_1,\ldots,j_n=0}^{N_j} h_{m,n}(k_1,\ldots,k_m,j_1,\ldots,j_n) \nonumber \\
						 & \: \: u(t-k_1)\ldots u(t-k_m) n_x(t-j_1) \ldots n_x(t-j_n), \nonumber
		\end{align}
		where $N_k$ and $N_j$ are the numbers of taps considered for the input $u$ and the process noise $n_x$ respectively.
		
		\begin{assumption} \label{ass:ProcNoise}
			\textbf{Process noise framework:}			
			A colored zero-mean noise source $n_x(t)$ is present as one of the inputs of the dual-input single-output Volterra representation of the nonlinear system.
			\begin{align}
				y_0(t) = \sum_{m=0}^{\infty} \sum_{n=0}^{\infty} V_{m,n}(u(t),n_x(t))
			\end{align}
			The noise $n_x(t)$ is assumed to be independent of the known input $u(t)$ and the output noise $n_y(t)$, its $n$-th order moments are finite $\forall \: n \in \mathbb{N}$.
		\end{assumption}
		
		\begin{thm} \label{theo:Volterra}
			If $y_{m,n}(t) = V_{m,n}(u(t),n_x(t))$, then $\bar{\bar{y}}_{m,n}(t) = \bar{\bar{V}}_{m,n}(u(t))$, where
			\begin{align}
				\bar{\bar{y}}_{m,n}(t) &= E_{n_x}\{ y_{m,n}(t)\}. \\
				\bar{\bar{V}}_{m,n}(u(t)) &= \sum_{k_1,\ldots,k_m=0}^{N_k} \bar{\bar{h}}_{m,n}(k_1,\ldots,k_m) \\
																	& \: \: u(t-k_1)\ldots u(t-k_m), \nonumber \\
				\bar{\bar{h}}_{m,n}(k_1,\ldots,k_m) &= \sum_{j_1,\ldots,j_n=0}^{N_j} h_{m,n}(k_1,\ldots,k_m,j_1,\ldots,j_n) \nonumber \\
																						& \: \: E_{n_x}\{ n_x(t-j_1) \ldots n_x(t-j_n) \} \label{eq:hbarbar}
			\end{align} 
		\end{thm}
		\begin{proof}
			The signal $\bar{\bar{y}}_{m,n}(t)$ is obtained by taking the expectation of \eqref{eq:ymn} w.r.t. the process noise $n_x$:
			\begin{align}
				\bar{\bar{y}}_{m,n}(t) &= E_{n_x}\{y_{m,n}(t)\} \\
															 &= \sum_{k_1,\ldots,k_m=0}^{N_k} \sum_{j_1,\ldots,j_n=0}^{N_j} h_{m,n}(k_1,\ldots,k_m,j_1,\ldots,j_n) \nonumber \\
						 			 						 & \: \: u(t-k_1)\ldots u(t-k_m) E_{n_x}\{n_x(t-j_1) \ldots n_x(t-j_n)\}, \nonumber
			\end{align}
			where the expectation $E_{n_x}\{n_x(t-j_1) \ldots n_x(t-j_n)\}$ depends on the $n$-th order moment of $n_x(t)$ and its decomposition into pairwise autocorrelations \citep{Schetzen1980}. This can be simplified using eq.~\eqref{eq:hbarbar}:
			\begin{align}
				\bar{\bar{y}}_{m,n}(t) &= \sum_{k_1,\ldots,k_m=0}^{N_k} \bar{\bar{h}}_{m,n}(k_1,\ldots,k_m) \\
															 & \: \: u(t-k_1)\ldots u(t-k_m), \nonumber \\
															 &= \bar{\bar{V}}_{m,n}(u(t)).
			\end{align}
			The sum in eq.~\eqref{eq:hbarbar} is finite since $h_{m,n}(k_1,\ldots,k_m,j_1,\ldots,j_n)$ is finite and it is assumed in Assumption~\ref{ass:ProcNoise} that the $n$-th order moments of the process noise $n_x(t)$ are finite. 
			
		\end{proof}	
		This theorem shows that the relation between $u$ and $y$ is given by a SISO Volterra representation after taking the expectation of $y$ with respect to the random realization of the process noise.
		
		\begin{mydef} \label{def:SystemProcess}
			$\mathbb{S}_p$ is the class of nonlinear systems such that, when excited with a random phase multisine, the following inequality holds:
			\begin{align}
				\exists \; C_1< \infty, \text{s.t.} \; \sum_{m=0}^{\infty}\sum_{n=0}^{\infty} M_{G^{m,n}}M^{m}_{U} \leq C_1,
			\end{align}
			with $M_{G^{m,n}} = \textit{max}\left|\bar{\bar{H}}^{m,n}_{L_k,k_1,k_2,\ldots,k_{m-1}}\right|$, where $\bar{\bar{H}}^{m,n}_{L_k,k_1,k_2,\ldots,k_{m-1}}$ is the symmetrized frequency domain representation of $\bar{\bar{h}}_{m,n}(k_1,\ldots,k_m)$.
		\end{mydef}
		Definition~\ref{def:SystemProcess} is the natural extension of the system class $\mathbb{S}$. Indeed when no process noise is present the system class $\mathbb{S}$ is obtained. The finite moment assumption on the process noise is also similar to the maximum amplitude restriction in Definition~\ref{def:System}.
		
	\subsection{Generalized Definition of the Best Linear Approximation}

		The BLA framework in the presence of process noise is defined in this section. The framework consists of four model components: the BLA $G_{bla}(q)$ itself, the stochastic nonlinear distortion $y_s(t)$ due to the randomized input, the process noise contribution $y_{p}(t)$ due to the process noise, and the output noise contribution $y_n(t)$ as depicted in Figure~\ref{fig:BestLinearApproxPN}. Note that the process noise contribution $y_p$ is a new constituent of the extended BLA framework due to the presence of process noise in the system.
				
		\begin{figure}
			\centering
				\includegraphics[width=0.95\columnwidth]{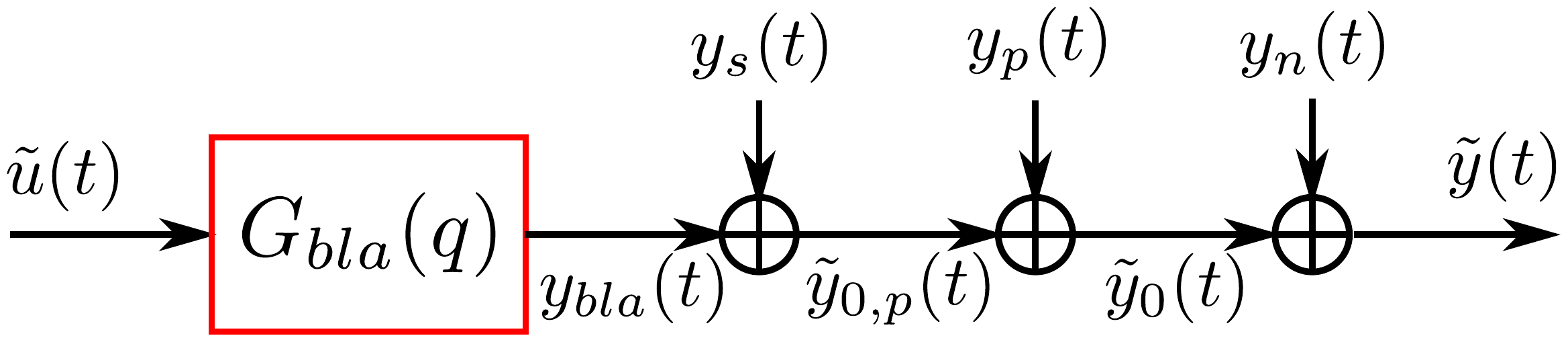}
			\caption{The BLA of a nonlinear system for a given class of excitation signals $\mathbb{U}$ consists of the resulting LTI model $G_{bla}(q)$, the unmodeled nonlinear contribution $y_{s}(t)$, the process noise distortion $y_{p}(t)$ and the output noise distortion $y_{n}(t)$.}
			\label{fig:BestLinearApproxPN}
		\end{figure}
		
		Two important design decisions are made in defining the extended BLA framework:	
		\begin{enumerate}
			\item The BLA $G_{bla}(q)$ and the stochastic nonlinear distortion $y_s$ are defined such that they do not depend on the actual realization of $n_x$ and $n_y$.
			\item The process noise contribution $y_p$ is defined such that it does not depend on the actual realization of the output noise $n_y$. 
		\end{enumerate}

		The BLA and the stochastic distortions $y_s$ are defined next as:
		\begin{align}
			G_{bla}(q) &= \underset{G(q)}{\argmin} \: E_{u,n_y,n_x}\left\{ \left| \tilde{y}(t) - G(q)\tilde{u}(t) \right|^{2} \right\}, \\
			y_{s}(t)   &= \bar{\bar{y}}(t) - y_{bla}(t) = \bar{\bar{y}}(t) - G_{bla}(q)\tilde{u}(t), \label{eq:ys}
		\end{align}
		where $G(q)$ belongs to the set of all possible LTI systems, and $\tilde{u}(t)$ is defined in eq.~\eqref{eq:uTilde}, $y_{bla}(t) = G_{bla}(q)\tilde{u}(t)$, $\bar{\bar{y}}(t)$ is defined as:
		\begin{align}
			\bar{\bar{y}}(t) &= E_{n_x,n_y}\{ \tilde{y}(t)\},
		\end{align}
		$\tilde{u}(t)$ and $\tilde{y}(t)$ are now defined as:
		\begin{align}
			\tilde{u}(t) &= u(t) - E_{u}\{ u(t)\}, \\
			\tilde{y}(t) &= y(t) - E_{u,n_x,n_y}\{ y(t)\}.
		\end{align}
 
 		The process noise contribution $y_p$ is defined as:
		\begin{align}
			y_p(t) &= \bar{y}(t) - \bar{\bar{y}}(t), \label{eq:procDef} \\
						 &= \tilde{y}_0(t) - G_{bla}(q)\tilde{u}(t) - y_{s}(t),
		\end{align}
		where $\bar{y}(t)$ is defined as:
		\begin{align}
			\bar{y}(t) &= E_{n_y}\{ \tilde{y}(t)\}, \\
								 &= \tilde{y}_0(t).
		\end{align}
		
		The output noise contribution $y_n$, the final constituent of the BLA framework, remains to be defined:
		\begin{align}
			y_n(t) &= \tilde{y}(t) - \bar{y}(t), \label{eq:outDef} \\
						 &= \tilde{y}(t) - \tilde{y}_0(t), \\
						 &= \tilde{y}(t) - G_{bla}(q)\tilde{u}(t) - y_{s}(t) - y_{p}(t), \\
						 &= n_y(t).
		\end{align}
		Note that, by definition, the noise contribution $y_{n}(t) = n_y(t)$. However, other choices for the definition of $y_{n}(t)$ may be made, resulting in different expressions of the noise contribution. This is highlighted in the last paragraph of Section~\ref{sec:Simulation}.
		
		To conclude we have that the system output $\tilde{y}(t)$, as shown in Figure~\ref{fig:BestLinearApproxPN}, is given by:
		\begin{align}
			\tilde{y}(t) &= y_{bla}(t) + y_s(t) \quad &\rightarrow & \: E_{n_x,n_y}\{ \tilde{y}(t)\} & \\
					 				 &\: \: + y_p(t)  					\quad &\rightarrow & \: E_{n_y}\{ \tilde{y}(t)\} - E_{n_x,n_y}\{ \tilde{y}(t)\} &\nonumber \\
					 				 &\: \: + y_n(t)  					\quad &\rightarrow & \: \tilde{y}(t) - E_{n_y}\{ \tilde{y}(t)\} &\nonumber
		\end{align}
		Section~\ref{sec:Estimation} presents how the BLA, and the variances of the signals $y_p(t)+y_n(t)$ and $y_s(t)+y_p(t)+y_n(t)$ can be estimated using the robust BLA estimation approach.
	
	\subsection{Properties of the BLA Model Components} \label{sec:properties}
		This section shows that the stochastic nonlinear distortion $y_s(t)$ and the process noise contribution $y_p(t)$ are zero-mean, and linearly uncorrelated with - but not independent of - the input $\tilde{u}(t)$.
		
		\begin{thm} \label{theo:zeromeanNL} Properties of the stochastic nonlinear distortion and the process noise contribution.  
			\begin{itemize}
			\item The stochastic nonlinear distortion $y_s(t)$ has zero-mean and is linearly uncorrelated with $\tilde{u}(t)$:
						\begin{align}
							E_u\{y_s(t)\} &= 0, \\
							E_u\{y_s(t)\tilde{u}(\tau)\} &= 0 \quad \quad \forall \: t,\tau.
						\end{align}
			\item The process noise contribution $y_p(t)$ has zero-mean and is linearly uncorrelated with $\tilde{u}(t)$:
						\begin{align}
							E_{n_x}\{y_p(t)\} &= 0, \\
							E_{n_x}\{y_p(t)\tilde{u}(\tau)\} &= 0 \quad \quad \forall \: t,\tau.
						\end{align}
			\item The sum of the process noise contribution $y_p(t)$ and the stochastic nonlinear contribution is uncorrelated with $\tilde{u}(t)$:
						\begin{align}
							E_{u,n_x}\{(y_p(t)+y_s(t))\tilde{u}(\tau)\} = 0 \quad \quad \forall \: t,\tau
						\end{align}
			\end{itemize}
		\end{thm}
		\begin{proof}
			The expected value of the stochastic nonlinear distortion $y_s(t)$ with respect to the input signal realization is given by:
			\begin{align}
				E_u\{y_s(t)\} = E_u\{\bar{\bar{y}}(t)\}-E_u\{G_{bla}(q)\tilde{u}(t)\}. 
			\end{align}
			The second term is equal to zero since $\tilde{u}(t)$ is zero-mean by construction. The first term is given by:
			\begin{align}
				E_u\{\bar{\bar{y}}(t)\} &= E_{u,n_x,n_y}\{ \tilde{y}(t)\},
			\end{align}
			where $E_{u}\{ \tilde{y}(t)\}$ is zero by construction. It hence follows directly from the definition of $y_s(t)$ that $E_u\{y_s(t)\} = 0$.

			The stochastic nonlinear distortion $y_s(t)$ is the residual of a linear least squares fit of a linear time-invariant model between $\bar{\bar{y}}(t)$ and $\tilde{u}(t)$ (see eq.~\eqref{eq:ys}). Hence $y_s(t)$ is linearly uncorrelated with $\tilde{u}(t)$ by construction in the absence of model errors ($G(q)$ belongs to the set of all possible LTI systems).

			The process noise contribution $y_p(t)$ is given by:
			\begin{align}
				y_p(t) &= \bar{y}(t) - \bar{\bar{y}}(t) \nonumber \\
				       &= E_{n_y}\{\tilde{y}(t)\} - E_{n_x,n_y}\{\tilde{y}(t)\}.
			\end{align}
			The expected value $E_{n_x}\{y_p(t)\}$ taken over the process noise realization is thus given by:
			\begin{align}
				E_{n_x}\{\bar{y}(t)\} &= E_{n_x,n_y}\{\tilde{y}(t)\} - E_{n_x,n_y}\{\tilde{y}(t)\} = 0.
			\end{align}

			The input $\tilde{u}$ does not depend on the process noise $n_x$ by construction and $E_{n_x}\{y_p(t)\}=0$ as is shown above. This results in:
			\begin{align}
				E_{n_x}\{y_p(t)\tilde{u}(\tau)\} &= E_{n_x}\{y_p(t)\}\tilde{u}(\tau) = 0. \label{eq:uncorPN}
			\end{align}
			
			It follows directly from the proof above that $E_{u,n_x}\{(y_p(t)+y_s(t))\tilde{u}(\tau)\} = 0 \: \forall \: t,\tau$.
		\end{proof}

		Many other properties of the BLA and its constituents can be proven based upon the assumption that the underlying nonlinear system is a Volterra system, and that the signal belongs to the Riemann equivalence class of asymptotically normally distributed excitation signals \citep{Pintelon2012}. Section~\ref{sec:VoltPN} showed that if the relationship between $\tilde{u}$, $n_x$  and $\tilde{y}$ is given by a Volterra system, then the relationship between $\tilde{u}$ and $\bar{\bar{y}}$ is also given by a Volterra system. As a consequence, the theoretical properties of $y_s$ and $G_{bla}$ proven for the basic case (see \citep{Pintelon2012} for an overview and a detailed analysis) still hold.
	
\section{The Best Linear Approximation in Feedback} \label{sec:FeedbackPN}

	\begin{figure}
		\centering
			\includegraphics[width=1\columnwidth]{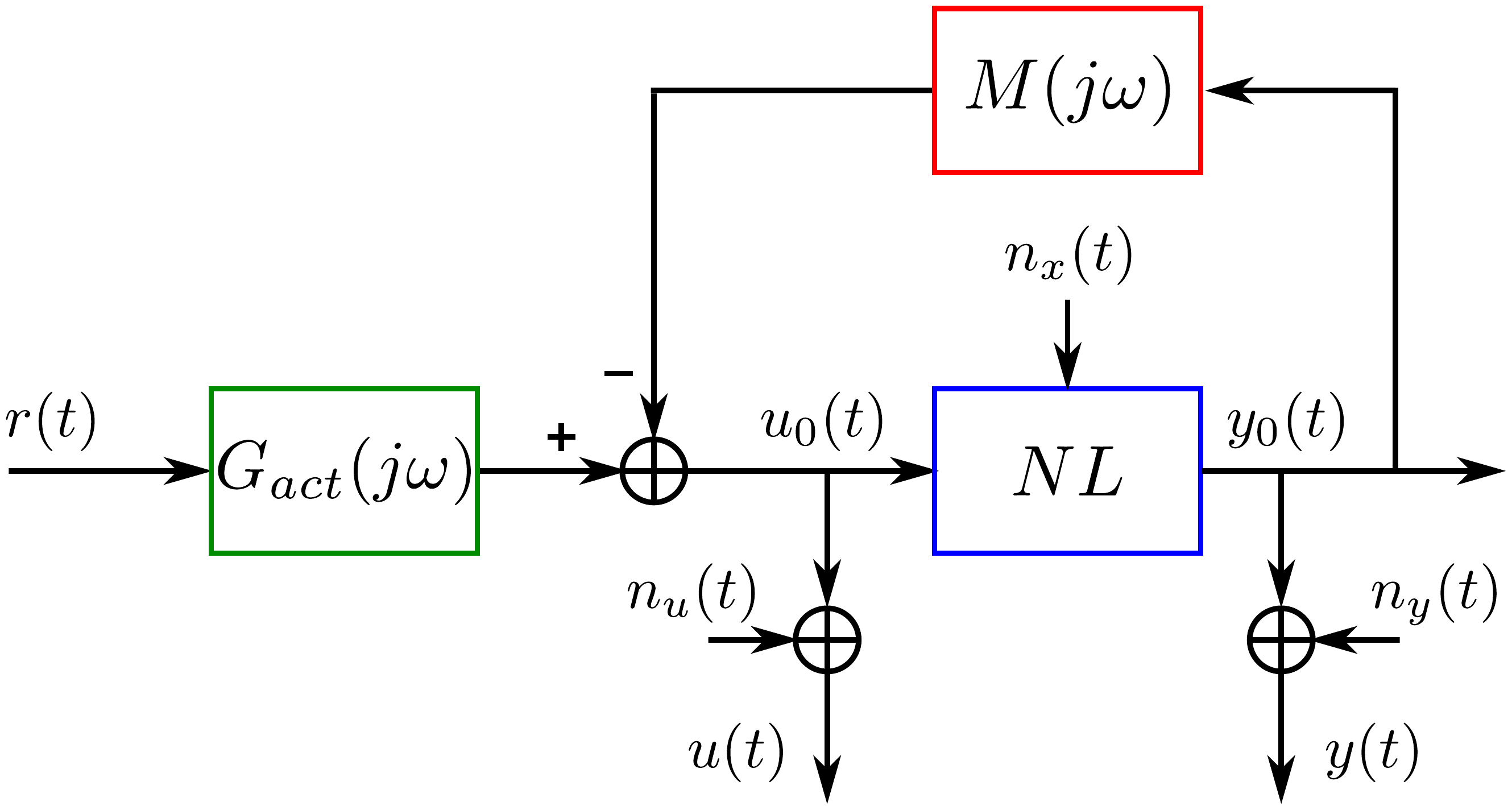}
		\caption{Setup for measuring the BLA of a nonlinear system with process noise operating in closed loop. The linear actuator and the feedback dynamics are represented by $G_{act}(j\omega)$ and $M(j\omega)$ respectively. $r(t)$ is the known reference signal, $u_0(t)$ and $y_0(t)$ are the noiseless input and output signal, $u(t)$ and $y(t)$ are the noisy input and output signal, and $n_u(t)$, $n_y(t)$ and $n_x(t)$ are the input measurement noise, output measurement noise and process noise respectively.}
		\label{fig:NonlinearFb}
	\end{figure}
		
	The classical open-loop BLA framework introduced in \citep{Schoukens1998} has been extended to systems operating in closed loop \citep{Pintelon2012a,Pintelon2013}. The generalized BLA applicable to the process noise presented in this paper is complementary with the closed-loop theory and can be similarly extended to closed loop systems. (see Figure~\ref{fig:NonlinearFb} for an overview of the setup).
	
	The closed-loop BLA is defined using the indirect frequency response function measurement method for linear feedback systems \citep{Wellstead1977,Wellstead1981}, it is based upon the open-loop relations from the reference signal to the system input and the system output. Not only the the relation from the reference signal $r(t)$ to the output signal $y(t)$, but also the relation from $r(t)$ to the input signal $u(t)$ is assumed to belong to the system class $\mathbb{S}_p$.
	
	Define:
	\begin{align}
		\bar{\bar{y}}(t) &= E_{n_x,n_y}\{ \tilde{y}(t)\}, \\
		\bar{\bar{u}}(t) &= E_{n_x,n_u}\{ \tilde{u}(t)\}.
	\end{align}
	Where $\tilde{u}(t)$ and $\tilde{y}(t)$ are now defined as:
	\begin{align}
		\tilde{u}(t) &= u(t) - E_{r,n_x,n_u}\{ u(t)\}, \\
		\tilde{y}(t) &= y(t) - E_{r,n_x,n_y}\{ y(t)\}.
	\end{align}
	
	The BLA of a nonlinear system with process noise operating in closed loop is now defined as:
	\begin{align}
		G_{bla}(j\omega) = \frac{S_{\bar{\bar{Y}}R}(j\omega)}{S_{\bar{\bar{U}}R}(j\omega)},
	\end{align}
	where $S_{\bar{\bar{Y}}R}(j\omega)$ and $S_{\bar{\bar{U}}R}(j\omega)$ are the reference-output and reference-input cross-power spectra respectively.
	
	The properties of the BLA of a system operating in feedback proven in \citep{Pintelon2012a,Pintelon2013} can easily be brought over to the process noise case using the properties proven in this paper.

\section{Estimating the BLA: the Robust Method} \label{sec:Estimation}

	The BLA $G_{bla}(j\omega)$ can be estimated both parametrically or nonparametrically, an extensive review of the available BLA estimation techniques is provided by \citep{Pintelon2012,Schoukens2016}. The presented methods remain valid in the process noise case. Some methods, such as the so-called robust method \citep{Pintelon2012,Schoukens2012,Schoukens2016}, can also provide an estimate of the noise variance and the total variance. This section first recapitulates the robust BLA estimation method, the behavior of the robust method is analyzed in detail for the process noise case next.
	
	\subsection{The Robust Method: Algorithm}
		The robust BLA estimation approach makes use of multiple periods and multiple realizations of a random phase multisine. The estimated BLA $\hat{G}_{bla}(j\omega)$ in open loop and with a known input is obtained as follows \citep{Pintelon2012,Schoukens2012,Schoukens2016}:
		\begin{align}
			\hat{G}^{[m,p]}(j\omega) &= \frac{Y^{[m,p]}(j\omega)}{U^{[m]}(j\omega)}, \\
			\hat{G}^{[m]}(j\omega)   &= \frac{1}{P} \sum_{p=1}^{P} \hat{G}^{[m,p]}(j\omega), \label{eq:GPer}\\ 
			\hat{G}_{bla}(j\omega) 	 &= \frac{1}{M} \sum_{m=1}^{M} \hat{G}^{[m]}(j\omega), \label{eq:GReal}
		\end{align}
	where $Y^{[m,p]}(j\omega)$ is the DFT of the $p$-th period and $m$-th realization of the output signal $y(t)$, $U^{[m]}(j\omega)$ is the $m$-th realization of the input signal. Since $u(t)$ is noise free, it is equal over all periods. The noise variance $\hat{\sigma}_{bla,n}^2(j\omega)$ (the variance on $\hat{G}_{bla}(j\omega)$ due to $y_n(t)$ in the output noise setting) and total variance $\hat{\sigma}_{bla,t}^2(j\omega)$ (the variance on $\hat{G}_{bla}(j\omega)$ due to $y_n(t) + y_s(t)$ in the output noise setting) estimate are given by:
		\begin{align}
			\hat{\sigma}_{bla,n}^2(j\omega) &= \frac{1}{M^2 P (P-1)} \sum_{m=1}^{M} \sum_{p=1}^{P} \left| \hat{G}^{[m]}(j\omega) - \hat{G}^{[m,p]}(j\omega) \right|^2, \label{eq:RobVarn} \\
			\hat{\sigma}_{bla,t}^2(j\omega) &= \frac{1}{M (M-1)} \sum_{m=1}^{M} \left| \hat{G}_{bla}(j\omega) - \hat{G}^{[m]}(j\omega) \right|^2. \label{eq:RobVart}
		\end{align}
		
		In order to quantify the variability of the mean BLA estimate \eqref{eq:GReal}, extra factors $M$ and $P$ have been introduced in the sample variances \eqref{eq:RobVarn} and \eqref{eq:RobVart}.
	
	\subsection{The Robust Method: Process Noise Analysis}
		The first step of the robust method takes the average of the output over the periods. Both the output noise contribution $y_n(t)$ and the process noise contribution $y_p(t)$ are aperiodic and zero-mean, while the stochastic nonlinear contribution is periodic. Hence, their contribution will be averaged out. In the second step, the average over the input realizations is taken. This step averages out the stochastic nonlinear contribution $y_s(t)$, but also the remaining contributions of the process noise and output noise $y_p(t)$ and $y_n(t)$. More formally we have that (see Figure~\ref{fig:BestLinearApproxPN}):
		\begin{align}
			&Y^{[m,p]}(j\omega) = \\
			&Y^{[m]}_{bla}(j\omega) + Y^{[m]}_{s}(j\omega) + Y^{[m,p]}_{p}(j\omega) + Y^{[m,p]}_{n}(j\omega), \nonumber
		\end{align}
		where $Y^{[m]}_{s}(j\omega)$, $Y^{[m,p]}_{p}(j\omega)$ and $Y^{[m,p]}_{n}(j\omega)$ are the DFT of the period $p$ and realization $m$ of the signals $y_{s}(t)$, $y_{p}(t)$ and $y_{n}(t)$ respectively.
		This results in the following expression for $\hat{G}^{[m,p]}(j\omega)$:
		\begin{align}
			&\hat{G}^{[m,p]}(j\omega)  \label{eq:BLAmp}\\
			&= \frac{Y^{[m]}_{bla}(j\omega)}{U^{[m]}(j\omega)} + \frac{Y^{[m]}_{s}(j\omega)}{U^{[m]}(j\omega)} + \frac{Y^{[m,p]}_{p}(j\omega) + Y^{[m,p]}_{n}(j\omega)}{U^{[m]}(j\omega)}, \nonumber \\
			&= G_{bla}(j\omega) + \frac{Y^{[m]}_{s}(j\omega)}{U^{[m]}(j\omega)} + \frac{Y^{[m,p]}_{p}(j\omega) + Y^{[m,p]}_{n}(j\omega)}{U^{[m]}(j\omega)}. \nonumber
		\end{align}
		
		A closer analysis, analogous to \citep{Pintelon2012}, of the expected value of eq.~\eqref{eq:GReal}, \eqref{eq:RobVarn}, \eqref{eq:RobVart} for the process noise case results in:
		\begin{align}
			E\{\hat{G}_{bla}(j\omega)\} &= G_{bla}(j\omega),\\
			E\{\hat{\sigma}_{bla,n}^2(j\omega)\} &=  \frac{\sigma^2_{n}(j\omega) + \sigma^2_{p}(j\omega)}{MP |U(j\omega)|^2},\\
			E\{\hat{\sigma}_{bla,t}^2(j\omega)\} &=  \frac{\sigma^2_{s}(j\omega)}{M |U(j\omega)|^2}+\frac{\sigma^2_{n}(j\omega) + \sigma^2_{p}(j\omega)}{MP |U(j\omega)|^2},
		\end{align}
		where $\sigma^2_{n}(j\omega)$, $\sigma^2_{p}(j\omega)$ and $\sigma^2_{s}(j\omega)$ are the variances of $Y^{[m,p]}_{n}(j\omega)$, $Y^{[m,p]}_{p}(j\omega)$, $Y^{[m]}_{s}(j\omega)$ respectively, and where $|U(j\omega)|^2$ is independent of the random phase realization. The expectations are taken with respect to the input realization, process noise realization and output noise realization.
		
		The robust BLA estimation method is still valid in the process noise case. However, the estimated BLA now depends on the process noise properties (see Section~\ref{sec:BlaPN}), and the estimated variance due to noise and the total variance on the BLA have an extra term which is process noise dependent. The variance due to the process noise contribution $y_p(t)$ and the output noise contribution $y_n(t)$ cannot be separated using the robust method. Note, however, that the presence of process noise in a nonlinear system can be detected using nonstationary input signals \citep{Zhang2017}.
	
\section{Example: A Hammerstein System} \label{sec:Simulation}
	
	\subsection{System}
		\begin{figure}
			\centering
				\includegraphics[width=1\columnwidth]{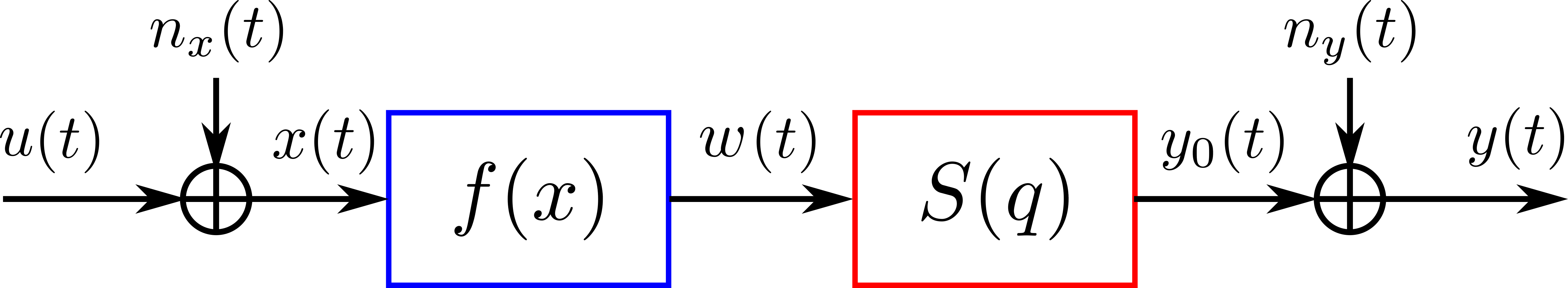}
			\caption{A Hammerstein system with process noise.}
			\label{fig:Example}
		\end{figure}
		
		Consider the following Hammerstein system (see Figure~\ref{fig:Example}), with  $f(x) = x + 0.1x^3$:
		\begin{align}
			y(t) &= S(q)\left[f(u(t)+n_x(t))\right] + n_y(t) \\
					 &= S(q)\left[u(t) + n_x(t) + 0.1u(t)^3 + 0.3u(t)^2 n_x(t)  \right. \nonumber \\
					 & \: \left. + 0.3u(t) n_x(t)^2 + 0.1n_x(t)^3\right] + n_y(t), 
		\end{align}
		Where $u(t)$, $n_x(t)$ and $n_y(t)$ are zero-mean white Gaussian signals with standard deviations $\sigma_u$, $\sigma_{n_x}$ and $\sigma_{n_y}$ respectively.
	
	\subsection{Theoretical Analysis}
		$\tilde{u}(t) = u(t)$ and $\tilde{y}(t) = y(t)$ since the input $u(t)$ has zero-mean and the static nonlinearity is odd. $\bar{y}(t)$ and $\bar{\bar{y}}(t)$ are given by:
		\begin{align}
			\bar{y}(t) &= E_{n_y}\{y(t)\}  \nonumber \\
								 &= S(q)\left[u(t) + n_x(t) + 0.1u(t)^3 + 0.3u(t)^2 n_x(t)  \right. \nonumber \\
								 & \: \: \left. + 0.3u(t) n_x(t)^2 + 0.1n_x(t)^3\right],  \\
			\bar{\bar{y}}(t) &= E_{n_y,n_x}\{y(t)\} \nonumber \\
											 &= S(q)\left[u(t) + 0.1u(t)^3 + 0.3u(t) \sigma_{n_x}^2\right].
		\end{align}	 
		
		Since the input of the static nonlinearity is Gaussian, Bussgang's Theorem can be applied \citep{Bussgang1952}, i.e. the BLA of a static nonlinearity is a static gain depending on the variance of the input $u(t)$ and the process noise $n_x(t)$. Based on the results in \citep{Enqvist2010,Giordano2016} we obtain:
		\begin{align}
			G_{bla}(q) &= S(q) (1 + 0.3 \sigma_u^2 + 0.3 \sigma_{n_x}^2).
		\end{align}
		The BLA constituents $y_{bla}(t)$, $y_{s}(t)$, $y_{p}(t)$, $y_{n}(t)$ are given by:
		\begin{align}
			y_{bla}(t) &= S(q)\left[(1 + 0.3 \sigma_u^2 + 0.3 \sigma_{n_x}^2)u(t)\right], \nonumber \\
			y_{s}(t) 	 &= S(q)\left[0.1u(t)^3 - 0.3\sigma_u^2 u(t)\right],\\
			y_{p}(t) 	 &= S(q)\left[n_x(t) + 0.3u(t)^2 n_x(t) \right. \nonumber \\
								 & \: \: \left. + 0.3u(t) (n_x(t)^2 - \sigma_{n_x}^2) + 0.1n_x(t)^3\right], \nonumber \\
			y_{n}(t) 	 &= n_{y}(t). \nonumber
		\end{align}
		It can easily be observed that the properties that are derived in Section~\ref{sec:properties} are valid for this case study. It can also be observed that the BLA does not only depend on the input signal properties, but also on the properties of the disturbing process noise (as it is also the case for the BLA in the feedback framework \citep{Pintelon2013}). It is illustrated in Figure~\ref{fig:BlaRobustChange}, for the Hammerstein case considered here, that the gain of the BLA depends on the variance of the process noise. The process noise contribution $y_p(t)$ on the other hand does not only depend on the process noise $n_x$(t), but also on the input signal $u(t)$.
		
		\begin{figure}
			\centering
				\includegraphics[width=1\columnwidth]{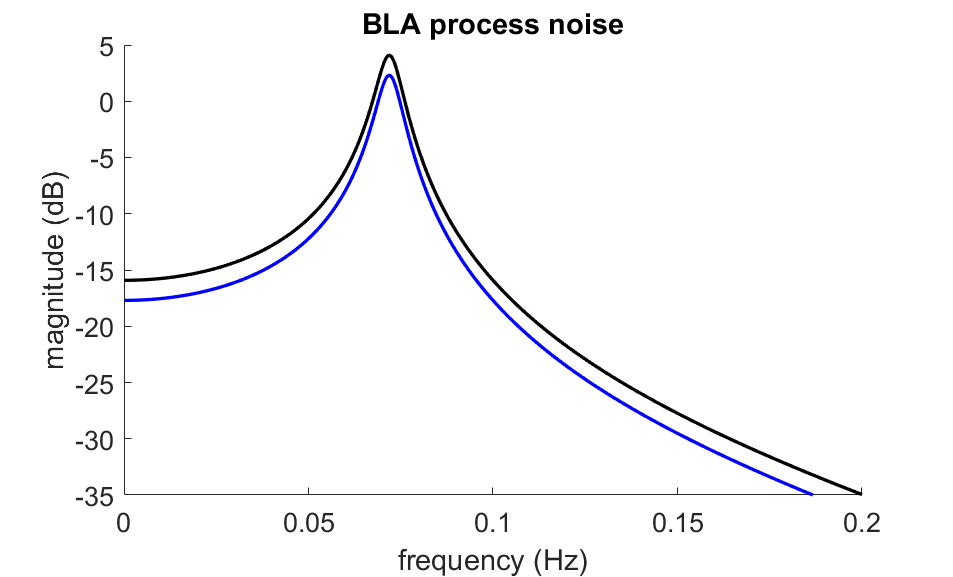}
			\caption{Dependency of the BLA on the process noise: the BLA is computed analytically with a process noise standard deviation $\sigma_{n_x}$ of 0.1 (blue) and 1 (black), while the input standard deviation $\sigma_{u}$ is fixed equal to 1. A clear gain increase can be observed.}
			\label{fig:BlaRobustChange}
		\end{figure}
		
		Note that the chosen definition of the process noise contribution $y_p$ and the output noise contribution $y_n$ are not unique (see eq. \eqref{eq:procDef} and \eqref{eq:outDef}). An alternative set of definitions $\check{y}_p$ and $\check{y}_n$ could be to assign all the noise terms depending on the input $u$ to the process noise contribution, and assign all the noise terms independent of the input $u$ to the output noise contribution, resulting in the following expressions for $\check{y}_p$ and $\check{y}_n$ in this example:
		\begin{align}
			\check{y}_{p}(t) &= S(q)\left[0.3u(t)^2 n_x(t) + 0.3u(t) (n_x(t)^2 - \sigma_{n_x}^2)\right], \nonumber \\
			\check{y}_{n}(t) &= n_{y}(t)+ S(q)\left[n_x(t) + 0.1n_x(t)^3\right].
		\end{align}
		However, the original definitions have the merit of being simple extension of the definitions used in the output noise open-loop and closed-loop setting, based on taking the expected value with respect to the output noise $n_y$ and the process noise $n_x$. Note as well that with the chosen definitions the output noise contribution only contains terms due to the output noise, while this is not the case using the alternative definition. For these reasons the authors have chosen to use the definitions that are expressed in eq. \eqref{eq:procDef} and \eqref{eq:outDef}. 
	
	\subsection{Robust Method Results}
		This section illustrates how the robust method can be used to estimate the BLA in a process noise setting. The experimentally obtained BLA, total variance and noise variance are compared with the analytically derived total and noise variance. Note that the robust approach does not require any knowledge of the system, while a full knowledge of the system and the (noise) signals distribution is required for the analytical derivation. A total of $M=10$ realizations is used, each containing 2 steady-state periods of 4096 points per period. The standard deviation of the input signal and noise signals are $\sigma_u = 1$, $\sigma_{n_x}=0.1$ and $\sigma_{n_y}=0.03$.
		
		The BLA $\hat{G}_{bla}(j\omega)$ and the variances $\hat{\sigma}_{bla,t}^2(j\omega)$, $\hat{\sigma}_{bla,n}^2(j\omega)$ obtained with the robust BLA estimation method coincide perfectly with their analytical counterparts as can be seen in Figure~\ref{fig:BlaRobust}. Note that the robust approach cannot distinguish between the process noise and the output noise variance, what is shown here is the total variance of the BLA due to both the process noise and the output noise.
	
		\begin{figure}
			\centering
				\includegraphics[width=1\columnwidth]{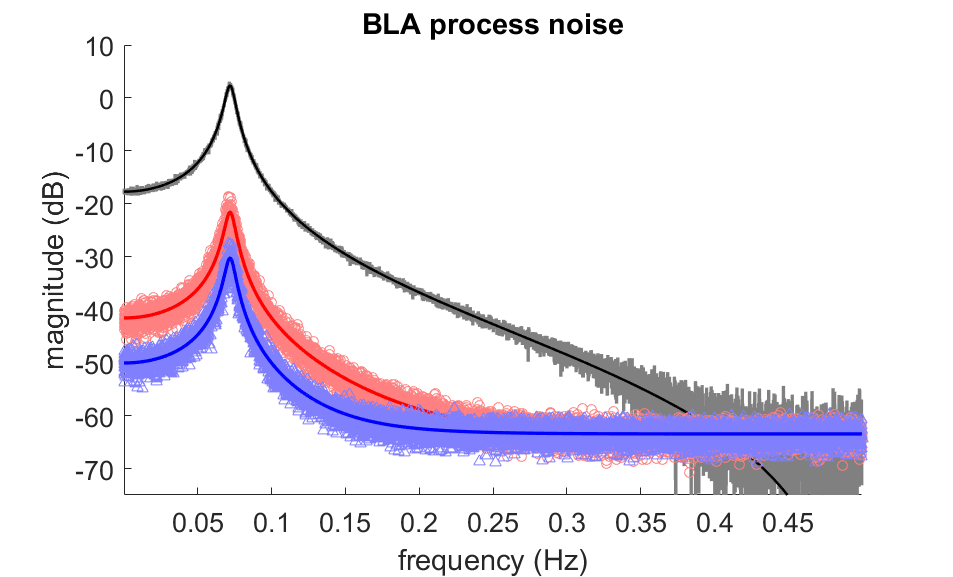}
			\caption{The estimated BLA versus the theoretically obtained BLA. The obtained BLA $G_{bla}(j\omega)$ is shown in black (analytical) and gray (robust method), the total variance on the estimated BLA is shown in light (robust method) and dark (analytical) red, while the noise variance is shown in light (robust method) and dark (analytical) blue. A good match between the analytical BLA expression and the one obtained experimentally using the robust method can be observed.}
			\label{fig:BlaRobust}
		\end{figure}

\section{Conclusion} \label{sec:conclusion}
	The Best Linear Approximation framework is extended to the process noise case, both for the open-loop and the closed-loop setting. The process noise acts as a second input of a Volterra system, resulting in a very general process noise framework. It is proven that the stochastic nonlinear contributions and the process noise contribution are zero-mean and uncorrelated with the input. It is also illustrated that the BLA can depend of the properties on the process noise, and that both the process noise contribution and the stochastic nonlinear distortion are uncorrelated but not independent of the input excitation. The Best Linear Approximation, together with the total and the noise variance can be obtained using the robust estimation method in the case of process noise.

\begin{ack} 
	 This work was supported in part by the Fund for Scientific Research (FWO-Vlaanderen), the Methusalem grant of the Flemish Government (METH-1), by the Belgian Government through the Inter university Poles of Attraction IAP VII/19 DYSCO program, and the ERC advanced grant SNLSID, under contract 320378. Maarten Schoukens is supported by the H2020 Marie Sklodowska-Curie European Fellowship. The project leading to this application has received funding from the European Union's Horizon 2020 research and innovation programme under the Marie Sklodowska-Curie grant agreement Nr 798627.
\end{ack}	

\bibliographystyle{plainnat}        													 
\bibliography{../../../References/ReferencesLibraryV2}  
		
\end{document}